\documentclass[11pt]{article}
\usepackage{amsmath,amssymb,amsbsy,amsfonts,amsthm,latexsym,
               amsopn,amstext,amsxtra,euscript,amscd,color}
\def \F {{\mathbb F}}
\def \Q {{\mathbb Q}}
\def \Z {{\mathbb Z}}

\newtheorem{theorem}{Theorem}[section]
\newtheorem{thm}{Theorem}[section]
\newtheorem{definition}[theorem]{Definition}
\newtheorem{lemma}[theorem]{Lemma}
\newtheorem{proposition}[theorem]{Proposition}

\newtheorem{corollary}[theorem]{Corollary}

\newcommand{\fun}[3]{{{#1}\,:\,{#2}\,\rightarrow\,{#3}}}
\newcommand{\zetak}[1][2^k]{\zeta_{{#1}}}

\newcommand{\Wa}[2][]{{\mathcal W}_{#2#1}}
\newcommand{\GF}[2][2]{{\mathbb F}_{{#1}^{#2}}}

\newcommand{\ZZ}[1]{{\mathbb Z}_{#1}}
\newcommand{\QQ}{{\mathbb Q}}

\def\cB{{\mathcal B}}
\def\cC{{\mathcal C}}

\def\cH{{\mathcal H}}

\def\cW{{\mathcal W}}

\def\cGB{\mathcal{GB}}

\def\aa  {{\bf a}}
\def\bb {{\bf b}}
\def\cc {{\bf c}}
\def\dd{{\bf d}}
\def\uu{{\bf u}}
\def\vv{{\bf v}}

\def\xx  {{\bf x}}
\def\yy{{\bf y}}
\def\zz{{\bf z}}

\def\00{{\bf 0}}
\def\11{{\bf 1}}
\def\+{\oplus}

\def \F {{\mathbb F}}
\def \Q {{\mathbb Q}}
\def \Z {{\mathbb Z}}

\newcommand{\BBF}{\mathbb{F}}

\newcommand\supp{{\rm supp}}

\begin{document}

\title{\huge\bf
Landscape Boolean functions}

\author{\Large  Constanza Riera$^1$,   Pantelimon St\u anic\u a$^2$\\
\vspace{0.4cm} \\
\small $^1$Department of Computing, Mathematics, and Physics,\\
\small Western Norway University of Applied Sciences \\
\small 5020 Bergen, Norway;  Email: {\tt csr@hvl.no}\\
\small $^2$Department of Applied Mathematics, \\
\small Naval Postgraduate School, \\
\small Monterey, CA 93943, U.S.A.; \small Email: {\tt pstanica@nps.edu}
}

\date{\today}
\maketitle

\begin{abstract}
In this paper we define a class of Boolean and generalized Boolean functions defined on $\F_2^n$ with values in $\Z_q$ (mostly, we consider $q=2^k$), which we call landscape  functions (whose class containing generalized bent, semibent, and plateaued)   and find their complete characterization in terms of their components. In particular, we show  that  the previously published characterizations of generalized bent and plateaued Boolean functions are in fact particular cases of this more general setting. Furthermore, we provide an inductive construction of landscape functions,  having any number of nonzero Walsh-Hadamard coefficients. We also completely characterize generalized plateaued functions in terms of the second derivatives and fourth moments.
\end{abstract}

\section{Introduction}

Generalized Boolean functions have become an active area of research \cite{hmp,HP,KSW85,MMMS17,MMS17,gbps,KUS1,ST09,smgs,Tok}, with most of these papers dealing with descriptions/constructions of generalized bent/plateaued functions. We show here that in fact these characterizations of generalized bent/plateaued in terms of the components of the function are in fact particular  instances of the more general case of generalized landscape functions, which are introduced in this paper.

We take $\F_2^n$ to be an $n$-dimensional vector space over the two-element field $\F_2$ and for an integer $q$, let $\Z_q$ be the ring of integers modulo $q$. By `$+$' and `$-$' we respectively denote addition
and subtraction modulo $q$,  while `$\oplus$' is the addition over $\BBF_2^n$. A  {\it generalized Boolean function} on $n$ variables is a function from $\F_2^n$ to $\Z_q$ ($q \geq 2 $), whose  set is denoted by   $\cGB_n^q$ and when $q=2$, by $\cB_n$. If $2^{k-1}<q \leq  2^k$ for some $k\ge 1$, we can associate to any $f \in \cGB_n^q$ a unique sequence of Boolean functions
$a_i\in \cB_n$ ($i=0,1,\ldots,k-1$) such that
\begin{equation*}
\label{eq0.1}
f(\xx) = a_0(\xx) + 2 a_1(\xx)+\cdots+2^{k-1} a_{k-1}(\xx), \mbox{ for all } \xx   \in \F_2^n.
\end{equation*}
The ({\em Hamming}) {\em weight} of $\xx   =(x_1,\ldots,x_n)\in \F_2^n$ is denoted by
$wt(\xx) $ and equals $ \sum_{i=1}^{n}x_i$
(the Hamming weight of a function is the weight of its truth table, that is, its output vector).
The cardinality, respectively, the complement  (in a universal set understood from the context)  of a set $S$ is denoted by $|S|$, respectively, $\overline{S}$.

For a {generalized Boolean function} $f:\F_2^n\to \Z_q$ we define the (unnormalized)  {\it generalized Walsh-Hadamard transform} to be the complex valued function
\[ \mathcal{H}^{(q)}_f(\uu) = \sum_{\xx  \in \F_2^n}\zeta_q^{f(\xx)}(-1)^{\uu\cdot \xx}, \]
where $\zeta_q = e^{\frac{2\pi i}{q}}$ is a $q$-complex root of $1$ and $\langle \uu ,\xx  \rangle$ denotes the conventional dot product  on $\F_2^n$
(for simplicity, we sometimes use $\zeta$, $\cH_f$, instead of $\zeta_q$, respectively, $\cH_f^{(q)}$,  when $q$ is fixed). The map $\displaystyle\mathcal{F}_f(\uu) = \sum_{\xx  \in \F_2^n} f(\xx) (-1)^{\langle \uu ,\xx  \rangle}$ is called the Fourier transform.
For $q=2$, we obtain the usual (unnormalized) {\it Walsh-Hadamard transform}
\[ 
\mathcal{W}_f(\uu) = \sum_{\xx  \in \F_2^n}(-1)^{f(\xx)}(-1)^{\uu\cdot \xx}.
\]

The sum $$\cC_{f,g}( \zz )=\sum_{\xx   \in \F_2^n} \zeta^{f(\xx  \+ \zz )  - g(\xx)}$$
is  the {\em crosscorrelation} of $f$ and
$g$ at $ \zz  \in \F_2^n$, and 
the {\em autocorrelation} of $f \in \cB_n$ at $ \uu  \in \BBF_2^n$
is $\cC_f(\uu):=\cC_{f,f}(\uu)$.
It is known (see~\cite{smgs}) that if $f,g\in\cGB_n^q$, then,
\begin{equation*}
\label{eq3}
\begin{split}
&\sum_{ \uu  \in \F_2^n}\cC_{f,g}(\uu)(-1)^{ \uu  \cdot \xx  } = \cH_f(\xx)\overline{\cH_g(\xx)}, \\
&\cC_{f,g}(\uu) = 2^{-n}\sum_{\xx   \in \F_2^n}\cH_f(\xx)\overline{\cH_g(\xx)} (-1)^{ \uu  \cdot \xx  },\\
&\cC_f(\uu) = 2^{-n}\sum_{\xx   \in \BBF_2^n}|\cH_f(\xx)|^2 (-1)^{ \uu  \cdot \xx  }.
\end{split}
\end{equation*}

A function $f:\F_2^n\rightarrow\Z_q$ is {\em generalized bent} ({\em gbent}) if $|\mathcal{H}_f(\uu)| = 2^{n/2}$, for all $ \uu \in \F_2^n$.
This is a generalization of functions $f$ for which $|\mathcal{W}_f(\uu)| = 2^{n/2}$, for all $ \uu \in \F_2^n$, which are called {\em bent} functions. 
In the spirit of Zheng and Zhang~\cite{zz}, we say that $f\in\mathcal{B}_n^q$ is  (generalized) {\em $s$-plateaued} if
$|\mathcal{H}_f(\uu)| \in \{0,2^{(n+s)/2}\}$ for all $ \uu \in \F_2^n$ for a fixed integer $s$ depending on $f$.
If $s=1$, or $s=2$, we call $f$ (generalized) {\em semibent}. See Mesnager's excellent survey~\cite{Mes14} for more on plateaued Boolean functions. Note that, for Boolean functions, bent functions exist only when $n$ is even, while this is not necessary when $q>2$. Similarly for semibent functions, as the conditions $n$ odd for $s=1$, and $n$ even for $s=2$ are no longer necessary when $q>2$.

Given a generalized Boolean function $f:\F_2^n\to\Z_q$, the {\em derivative} $D_{\aa  }f$ of $f$ with respect to a vector~$\aa  $ is the generalized Boolean function $D_\aa f:\F_2^n\to\Z_q$ defined by
\begin{equation*}
 D_{a  }f(\xx) = f(\xx)- f(\xx  \oplus \aa  ), \mbox{ for  all }  \xx   \in \F_2^n.
\end{equation*}
Certainly, if $f$ is Boolean, then $D_{a  }f(\xx) = f(\xx)\+ f(\xx  \oplus \aa  )$.
For $f\in\mathcal{GB}_n^q$, we let $\supp(\cH_f)=\{\uu\,:\, \cH_f(\uu)\neq 0\}$, the {\em spectra support} of $f$. In this paper, we consider the case $q=2^k$ (of course, one can easily write the results for $2^{k-1}<q\leq 2^k$).
 
\section{Landscape functions and their regularity}

As defined in~\cite{MMS17}, a  gbent function $f\in\mathcal{GB}_n^q$  is {\em regular}, if $\mathcal{H}_f(\uu) = 2^{n/2}\zeta_q^{f^*(\uu)}$ for some function $f^*\in\mathcal{GB}_n^q$, called the dual. 
We extend this definition in the following way: we call a function $f\in\mathcal{GB}_n^q$  {\em regular}, if for all $\uu\in\supp(\cH_f)$, $\mathcal{H}_f(\uu) = 2^{\frac{n_\uu}{2}}\ell_\uu\,\zeta_q^{f^*(\uu)}$, for some $n_\uu \in\mathbb{Z}^{\geq 0}$, $\ell_\uu\in2\mathbb{Z}^{\geq 0}+1$ and some $f^*(\uu)\in\Z_q$. We define a function $f^*\in\mathcal{GB}_n^q$ by extending these values  outside of the spectra  support of $f$ by $f^*(\uu)=0$, for all $\uu\in\overline{\supp(\cH_f)}$, and call this function, the {\em dual} of $f$ (note: the function $f$ cannot be recovered from $f^*$, in general).

By modifying a method of Kumar, Scholtz and Welch~\cite{KSW85}, in~\cite{MMS17} it was shown that
all gbent functions $f\in \cGB_n^{2^k}$ are regular, except for $n$ odd and $k=2$, in which case one has
$\displaystyle \mathcal{H}_f(\uu) = 2^{\frac{n-1}{2}}(\pm 1\pm i)$.
We observe that with our definition of regularity, a function cannot be regular unless the moduli of all  nonzero Walsh-Hadamard coefficients of $f$ are of the form $2^{\frac{m_1}{2}}\ell_1$, $2^{\frac{m_2}{2}}\ell_2,\ldots$  with $m_1,m_2,\ldots\in\mathbb{Z}^{\geq 0}$, $\ell_1,\ell_2,\ldots \in2\Z^{\geq 0}+1$. With that in mind, we introduce the following notion.

\begin{definition}
We call a function $f\in\cGB_n^{q}$ a {\em landscape} function if there exist $t\geq 1$, $m_i\in\Z^{\geq 0},\ell_i\in2\Z^{\geq 0}+1$, $1\leq i\leq t$, such that
\[
\{|\cH_f(\uu)|\}_{\uu\in\supp(\cH_f)}=\{2^{\frac{m_1}{2}}\ell_1,\ldots,2^{\frac{m_t}{2}}\ell_t\}.
\]
 We call the set of pairs $\{(m_1,\ell_1),(m_2,\ell_2),\ldots\}$, the {\em levels}  of $f$, and  $t+1$ (if $0$ belongs to the Walsh-Hadamard spectrum), or $t$ (if $0$ is not in the spectrum) the {\em length} of   $f$.
\end{definition}
Certainly, every classical Boolean function is a landscape function. That is not true for $q>2$ (as the Walsh-Hadamard values are $\pm$ sums of powers of the primitive root, so the moduli of the spectra values may contain elements outside $\Z\cup\sqrt{2}\Z$), however, generalized bent, semibent, or more general, plateaued functions are all examples of landscape  functions.

In Theorem~\ref{constr-ind} we will construct, in an inductive fashion, large classes of (generalized) landscape functions $:\F_2^n\to \Z_{2^k}$, for all $k\geq 2$.

First, we show the regularity of (generalized) landscape functions, by
modifying the proof from~\cite{KSW85,MMS17}.  Recall that when $q=2^k$ is fixed, we use $\cH_f$, in lieu of $\cH_f^{(2^k)}$. 
 \begin{theorem}
\label{reg-plat}
Let $f\in\mathcal{GB}_n^q$, $q=2^k$, $k\geq 1$,   be a landscape  function, and $\zeta=e^{\frac{2\pi i}{2^k}}$ be a $2^k$-primitive root of unity.  Then, for all $\uu\in\supp(\cH_f)$, with $|\cH_f(\uu)|=2^{\frac{m}{2}}\ell$, $m \in\mathbb{Z}^{\geq 0}$, $\ell\in2\Z^{\geq 0}+1$, we have $\cH_f(\uu)=2^{\frac{m}{2}}\ell\zeta^{f^*(\uu)}$,  for some value $f^*(\uu)\in\Z_q$, except for $m$ odd and $k=2$, in which case,  $\displaystyle \mathcal{H}_f(\uu) = 2^{\lfloor{m/2}\rfloor}\left(\ell_1\epsilon_1+\ell_2\epsilon_2\pm i  (\ell_1\epsilon_2-\ell_2\epsilon_1)\right)$, for $\epsilon_1,\epsilon_2\in\{\pm 1\}$, if $\ell$ is the largest component of a Pythagorean triple $\ell_1^2+\ell_2^2=\ell^2$, or $\displaystyle \mathcal{H}_f(\uu) = 2^{\lfloor{m/2}\rfloor} (\pm 1\pm i)$, if $\ell$ is not the largest component of a Pythagorean triple. 
\end{theorem}
\begin{proof}
If $k=1$, the result simply states that if $\uu\in\supp(\cW_f)$ and $|\cW_f(\uu)|= 2^{\frac{m}{2}}\ell$, then $\cW_f(\uu)=2^{\frac{m}{2}} \ell (-1)^{f^*(\uu)}$, which is certainly true, since the two roots of unity are $\pm 1$. 

Let $k\geq 2$, $\uu\in\supp(\cH_f)$ with $|\cH_f(\uu)|=2^{\frac{m}{2}}\ell$ (recall that $m,\ell \in\Z^{\geq 0}$, $\ell$ odd), and assume that $m$ is even, or, $m$ is odd and $k\neq 2$.  
 As in~\cite{KSW85,MMS17}, the ideal generated by $2$ is totally ramified in $\Z[\zeta]$ (which is the ring of algebraic integers in the cyclotomic field~$\Q(\zeta)$), so we have the decomposition  in $\Z[\zeta]$  of the ideal $\displaystyle \langle 2\rangle= \langle 1-\zeta\rangle^{2^{k-1}}$, where $\langle 1-\zeta\rangle$ is a prime ideal in $\Z[\zeta]$. 
 Observe that $\cH_f(\uu)\overline{\cH_f(\uu)}=2^{m}\ell^2$. From~\cite[Property 7]{KSW85}, we observe 
 that 
 $\cH_f(\uu)$ and $\overline{\cH_f(\uu)}$ will generate the same ideal in
 $\Z[\zeta]$ and so, $2^{-m}\ell^{-2}(\cH_f(\uu))^2$ is a unit, and consequently, $2^{-{\frac{m}{2}}} \ell^{-1}\cH_f(\uu)$ is an algebraic integer.  Therefore, by Proposition 1 of~\cite{KSW85},
 $2^{-{\frac{m}{2}}}\ell^{-1} \cH_f(\uu)$ is a root of unity.  
 Further, observe that the Gauss quadratic sum
 $\displaystyle G(2^k)=\sum_{i=0}^{2^k-1} \zeta^{i^2}= 2^{k/2}(1+i)$ and so, $\sqrt{2}\in\Q(\zeta)$, and so the root of unity  $2^{-{\frac{m}{2}}} \ell^{-1}\cH_f(\uu)$ must be in the  cyclotomic field $\Q(\zeta)$, unless $k=2$ (since then ${1+i}\not\in \Q(\zeta)$).  The first assertion is shown for $m$ even, as well as for $m$ odd with $k\neq 2$.

When $m$ is odd and $k=2$,   then $\cH_f(\uu)=a_\uu+b_\uu\,i$,
for some integers $a_\uu, b_\uu$. Since $|\cH_f(\uu)|^2=2^{m}\ell^2$,  we get the diophantine
equation $a_\uu^2 + b_\uu^2 = 2^{m}\ell^2$. 
We use now the fact that the product of sums of squares is a sum of squares, that is,
\[
(a^2+b^2)(c^2+d^2)=(ac+bd)^2+(ad-bc)^2.
\]
Since $m$ is odd, the solutions for $x^2+y^2=2^m$ are $(x,y)=(\epsilon_1 2^{\lfloor{m/2}\rfloor},\epsilon_2 2^{\lfloor{m/2}\rfloor})$, with $\epsilon_1,\epsilon_2\in\{\pm 1\}$, and if $\ell$ is the largest component of a Pythagorean triple $(\ell_1,\ell_2,\ell)$, all solutions of $a_\uu^2 + b_\uu^2 = 2^{m}\ell^2$ are of the form
$(a_\uu, b_\uu) = (2^{\lfloor{m/2}\rfloor}(\ell_1\epsilon_1+\ell_2\epsilon_2),\pm 2^{\lfloor{m/2}\rfloor}(\ell_1\epsilon_2-\ell_2\epsilon_1))$,  
and the theorem is shown.
\end{proof}




  \section{Characterizing landscape   functions in terms of components}
\label{main-comp}

In this section, we will completely characterize the landscape   functions in terms of their components, by using the method of~\cite{MMMS17}. It is rather intriguing that bentness does not play a role in the method rather the modulus of  values of the Walsh-Hadamard spectrum  being of the form $2^{\frac{m}{2}}\ell$ is important.

We define the ``canonical bijection''  $\fun{\iota} {\F_2^{k-1}}{\ZZ{2^{k-1}}}$  by $\iota(\cc) = \sum_{j=0}^{k-2}c_j 2^j$ where $\cc=(c_0,c_1,\dots,c_{k-2})$. 
We gather in the next lemma some computations from~\cite{MMMS17,smgs,txqf}, providing a relationship between the generalized Walsh-Hadamard transform and the classical transform.
\begin{lemma}
\label{lem:H-W}
For a generalized Boolean $f\in\cGB_n^{2^k}$, $f(\xx) = a_0(\xx)+2a_1(\xx)+\cdots+ 2^{k-1}a_{k-1}(\xx)$, $a_i\in\cB_n$, we have
\begin{align*}
\label{eq:relationGWaGray}
  \cH_{f} (\aa)
  &= \frac{1}{2^{k-1}} \sum_{(\cc,\dd)\in\F_ {2}^{k-1}\times\F_{2}^{k-1}} (-1)^{\cc\cdot \dd}\zetak^{\iota(\dd)}\,    \cW_{f_\cc}(\aa),
\end{align*}
where $f_\cc(\xx) = c_1a_1(\xx)\+c_2a_2(\xx)\+\cdots\+c_{k-2}a_{k-2}(\xx)\+a_{k-1}(\xx)$.
\end{lemma}
 
\begin{theorem}
\label{maj-thm}
Let $f:\F_2^n\rightarrow\Z_{2^k}$, $k\geq 2$,  be a  function given as
$f(\xx) = a_0(\xx)+2a_1(\xx)+\cdots+ 2^{k-1}a_{k-1}(\xx)$.  
Then, $f$ is a landscape   function whose  spectra moduli are in $\left\{0,2^{\frac{m_1}{2}}L_1,\ldots,2^{\frac{m_t}{2}}L_t\right\}$ ($t\in\Z^{>0},m_i\in\Z^{\geq 0}$, $L_i\in2\Z^{\geq 0}+1$) if and only if for each $\cc\in\F_2^{k-1}$, the Boolean function $f_\cc$ defined as
\[ f_\cc(\xx) = c_1a_1(\xx)\+c_2a_2(\xx)\+\cdots\+c_{k-2}a_{k-2}(\xx)\+a_{k-1}(\xx) \]
is a Boolean   function such that (we take $\ell\in\{L_1,\ldots,L_t\}$): 
\begin{itemize}
\item[$(i)$]   $\cH_f(\aa)=0$, if and only if $ \Wa{f_\cc}(\aa)=0$.
\item[$(ii)$] $|\cH_f(\aa)|=2^{\frac m2}\ell$, $m$ even, if and only if $ \Wa{f_\cc}(\aa)= (-1)^{\cc\cdot \iota^{-1}( g(\aa))+s(\aa)}2^{\frac{m}{2}}\ell$, for some
$\fun {g}{\F_2^n}{\ZZ{2^{k-1}}}$, $\fun {s}{\F_2^n}{\F_{2}}$.
\item[$(iii)$] $|\cH_f(\aa)|=2^{\frac m2}\ell$, $m$ odd, $k\neq2$, if and only if 
\[
\Wa{f_\cc}(\aa)= \left((-1)^{\cc\cdot \iota^{-1}(g_1(\aa))+s_1(\aa)}-(-1)^{\cc\cdot \iota^{-1}(g_2(\aa))+s_2(\aa)}\right)\,2^{\lfloor\frac {m}{2}\rfloor}\ell, 
\]
 for some $g_j:\F_{2}^n\to \Z_{2^{k-1}}, s_j:\F_{2}^n\to\F_2$, $j=1,2$,
where $g_2(\aa)-g_1(\aa)+2^{k-1}(s_2(\aa)-s_1(\aa))=2^{k-2}$ in $\Z_{2^k}$.
\item[$(iv)$] $|\cH_f(\aa)|=2^{\frac m2}\ell$, $m$ odd, $k=2$, if and only if   (note that $c$ is a bit)
\begin{align*}
\Wa{f_c}(\aa)=
\begin{cases} 
2^{\frac{m-1}{2}}\left(\ell_1\epsilon_1+\ell_2\epsilon_2\pm(-1)^c (\ell_1\epsilon_2-\ell_2\epsilon_1) \right), &\text{ if }  \ell_1^2+\ell_2^2=\ell^2\\
 2^{\frac{m-1}{2}}\ell\,\left(\pm1\pm(-1)^c\right), &\text{ otherwise}.
\end{cases}
\end{align*}
\end{itemize} 
Consequently, $f_\cc$ has nonzero spectra moduli given by $\left\{ 2^{\lceil{\frac{m_1}{2}}\rceil}L_1,\ldots,2^{\lceil{\frac{m_t}{2}}\rceil}L_t\right\}$.
\end{theorem}
\begin{proof}
$(i)$ First, let us treat the case of $\aa\notin\supp(\cH_f)$. Thus, 
\[
0=  \sum_{(\cc,\dd)\in\F_ {2}^{k-1}\times\F_{2}^{k-1}} (-1)^{\cc\cdot \dd}\zetak^{\iota(\dd)}\,    \cW_{f_\cc}(\aa)=\sum_{\dd\in\F_2^{k-1}} \left(\sum_{\cc\in\F_2^{k-1}} (-1)^{\cc\cdot \dd}\cW_{f_\cc}(\aa) \right)  \zetak^{\iota(\dd)},
\]
and so, $\displaystyle \sum_{\cc\in\F_2^{k-1}} (-1)^{\cc\cdot \dd}\cW_{f_\cc}(\aa)=0$, 
since $\{1,\zetak,\dots,\zetak^{2^{k-1}-1}\}$ is a basis of $\QQ(\zetak)$. 
Inverting, we get
\begin{align*}
  \cW_{f_\cc}(\aa)= \frac{1}{2^{k-1}}  \sum_{(\uu,\vv)\in\F_ 2^{k-1}} (-1)^{(\uu + \cc)\cdot \vv}
   \cW_{f_{\uu}}(\aa)=0.
\end{align*}
The converse follows easily.

$(ii)$ Let now $\aa\in\supp(\cH_f)$ with $|\cH_f(\aa)|=2^{\frac{m}{2}}\ell$, $m\in\Z$ even, $\ell$ odd.
Then,  $\fun f{\F_2^n}{\ZZ{2^k}}$ satisfies
$
  \cH_{f}(\aa) = 2^{\frac m2}\ell\,\zetak^{f^\star(\aa)}
$
for some $\fun {f^\star}{\F_2^n}{\ZZ{2^k}}$. Decompose $f^\star$ as $f^\star = g+2^{k-1}s$ with
$\fun {g}{\F_2^n}{\ZZ{2^{k-1}}}$ and $\fun {s}{\F_2^n}{\F_{2}}$ so that
\begin{align*}
  \cH_{f}(\aa) = 2^{\frac m2}\ell\,(-1)^{s(\aa)}\zetak^{g(\aa)}.
\end{align*}
Then,
\begin{equation}
\begin{split}
\label{eq:vanishing1}
&  \sum_{\dd\in\F_{2}^{k-1}}\left( \frac{1}{2^{k-1}} \sum_{\cc\in\F_ {2}^{k-1}} (-1)^{\cc\cdot \dd}
  \cW_{f_\cc}(\aa)\right)\zetak^{\iota(\dd)}  - 2^{\frac m2}\ell(-1)^{s(\aa)}\zetak^{g(\aa)} = 0.
\end{split}
\end{equation}

Again using that $\{1,\zetak,\dots,\zetak^{2^{k-1}-1}\}$ is a basis of $\QQ(\zetak)$ (denoting by $\delta_0$ the Dirac symbol $\delta_0(u,v)=1$  if $u=v$, and 0, otherwise), we infer
\begin{equation}
\label{eq:valeur1}
  \frac{1}{2^{k-1}} \sum_{\cc\in\F_2^{k-1}} (-1)^{\cc\cdot \dd}
  \cW_{f_\cc}(\aa) = 2^{\frac m2}\ell\,(-1)^{s(\aa)} \delta_0\left( \iota(\dd),g(\aa)\right).
\end{equation}

We now invert the above identity,  so, for any $\cc\in\F_2^{k-1}$,
\begin{align*}
  \cW_{f_\cc}(\aa)
  &= \frac{1}{2^{k-1}}  \sum_{(\uu,\vv)\in\F_ 2^{k-1}} (-1)^{(\uu + \cc)\cdot \vv}
   \cW_{f_{\uu}}(\aa)\\
  &=\sum_{\vv\in\F_ 2^{k-1}}(-1)^{\cc\cdot \vv}\left(\frac{1}{2^{k-1}}
    \sum_{\uu\in\F_ 2^{k-1}} (-1)^{\uu\cdot \vv}
   \cW_{f_\uu}(\aa)\right)\\
  &= (-1)^{\cc\cdot \iota^{-1}(g(\aa))+s(\aa)}2^{\frac m2}\ell.
\end{align*}
 That shows that $f_\cc$ satisfies the imposed conditions on the Walsh-Hadamard coefficient at $\aa$.

Conversely, suppose that there exist
$\fun {g}{\F_2^n}{\ZZ{2^{k-1}}}$ and $\fun{s}{\F_2^n}{\F_{2}}$
such that, for every $\cc\in\F_{2}^{k-1}$,
$
 \cW_{f_\cc}(\aa) = 2^{\frac m2}\ell(-1)^{\cc\cdot \iota^{-1}(g(\aa))+s(\aa)}
$.
By Lemma~\ref{lem:H-W},  we can write
\allowdisplaybreaks
\begin{align*}
  &\cH_{f} (\aa)
  = \frac{1}{2^{k-1}} \sum_{(\cc,\dd)\in\F_ 2^{k-1}\times\F_ 2^{k-1}} (-1)^{\cc\cdot \dd}\zetak^{\iota(\dd)}
   \cW_{f_\cc}(\aa)\\
  =& 2^{\frac m2}\ell\cdot\frac{1}{2^{k-1}}
    \sum_{(\cc,\dd)\in\F_2^{k-1}\times\F_ 2^{k-1}} (-1)^{\cc\cdot \dd+\cc\cdot \iota^{-1}(g(\aa))+s(\aa)}\zetak^{\iota(\dd)}\\
  = &2^{\frac m2} \ell\,(-1)^{s(\aa)}
    \sum_{\dd\in\F_ 2^{k-1}}\left(\frac{1}{2^{k-1}}\sum_{\cc\in\F_{2}^{k-1}}(-1)^{\cc\cdot (\dd+ \iota^{-1}(g(\aa)))}\right)\zetak^{\iota(\dd)}\\
   =&  2^{\frac m2} \ell\,(-1)^{s(\aa)} \zetak^{g(\aa))},
\end{align*}
proving that $f$ satisfies $|\cH_f(\aa)|=2^{\frac{m}{2}}\ell$.

$(iii)$ Now, let $\aa\in\supp(\cH_f)$, with $|\cH_f(\aa)|=2^{\frac{m}{2}}\ell$, $m,\ell\in\Z^{\geq 1}$ odd, $k\neq2$.
By Theorem~\ref{reg-plat}, then    
  \begin{align*}
    \cH_{f}(\aa)
     &= 2^{\frac m2}\ell\,\zetak^{f^\star(\aa)}
       = 2^{\frac{m-1}2}\, \sqrt{2}\, \ell\,\zetak^{f^\star(\aa)},
  \end{align*}
  for some power ${f^\star}(\aa) \in{\ZZ{2^k}}$.  Recall now that
  $\mathbb Q(\sqrt 2)\subset\mathbb Q (\zetak)$, since
  $\sqrt 2=\zetak[8] +\bar{\zetak[8]}=\zetak[8]-\zetak[8]^3=\zetak^{2^{k-3}}-\zetak^{3\cdot
    2^{k-3}}$. Thus,
\begin{align*}
  \cH_{f}(\aa)
  &=  2^{\frac{m-1}2}\ell\,\left( \zetak^{f^\star(\aa)+2^{k-3}} - \zetak^{f^\star(\aa)+3\cdot 2^{k-3}}\right).
\end{align*}
As in~\cite{MMMS17}, we let $f^\star(\aa)+2^{k-3}=g_1(\aa)+2^{k-1}s_1(\aa) + 2^kt_1(\aa)$ and
$f^\star(\aa)+3\cdot 2^{k-3}=g_2(\aa)+2^{k-1}s_2(\aa)+2^k t_2(\aa)$, where $g_i:\F_{2^n}\rightarrow\Z_{2^{k-1}}$ and $s_i, t_i:\F_{2^n}\rightarrow\F_2$, so that
\begin{align}
\label{eq:m-odd}
  \cH_{f}(\aa) = 2^{\frac{m-1}2}\ell\,(-1)^{s_1(\aa)}\zetak^{g_1(\aa)} -2^{\frac{m-1}2}\ell\,(-1)^{s_2(\aa)} \zetak^{g_2(\aa)}.
\end{align}

Recall that $\iota$ is the
canonical bijection from $\F_2^{k-1}$ to $\ZZ{2^{k-1}}$,  $\iota(c_0,\dots,c_{k-2})=\sum_{j=0}^{k-2}c_j 2^j$.
Using Lemma~\ref{lem:H-W}, we write
\begin{align*} 
  \cH_{f} (\aa)
  &= \frac{1}{2^{k-1}} \sum_{(\cc,\dd)\in\F_ {2}^{k-1}\times\F_{2}^{k-1}} (-1)^{\cc\cdot \dd}\zetak^{\iota(\dd)}
    \Wa{f_\cc}(\aa)\\
  &= \sum_{\dd\in\F_2^{k-1}} \left( \frac{1}{2^{k-1}} \sum_{\cc\in\F_ 2^{k-1}} (-1)^{\cc\cdot \dd}
    \Wa{f_\cc}(\aa) \right)\zetak^{\iota(\dd)},
\end{align*}
which combined with equation~\eqref{eq:m-odd}, renders,
\begin{align*}
  \label{eq:valeur2}
  \frac{1}{2^{k-1}} \sum_{\cc\in\F_ 2^{k-1}} (-1)^{\cc\cdot \dd}
  \Wa{f_\cc}(\aa)
  &= 2^{\frac {m-1}2}\ell\,(-1)^{s_1(\aa)} \delta_0\left(\iota(\dd),g_1(\aa)\right)\\
  &\qquad  -2^{\frac {m-1}2}\ell\,(-1)^{s_2(\aa)}   \delta_0\left(\iota(\dd),g_2(\aa)\right).
\end{align*}
Thus,
\begin{align*}
  \Wa{f_\cc}(\aa)
  &= \frac{1}{2^{k-1}}  \sum_{(\uu,\vv)\in\F_ 2^{k-1}\times\F_ 2^{k-1}} (-1)^{(\uu + \cc)\cdot \vv}
    \Wa{f_\uu}(\aa)\\
  &=\sum_{\vv\in\F_ 2^{k-1}}(-1)^{\cc\cdot \vv}\frac{1}{2^{k-1}}
    \sum_{\uu\in\F_ 2^{k-1}} (-1)^{\uu\cdot \vv}
    \Wa {f_\uu}(\aa)\\
  &= \frac{(-1)^{\cc\cdot \iota^{-1}(g_1(\aa))+s_1(\aa)}-(-1)^{\cc\cdot \iota^{-1}(g_2(\aa))+s_2(\aa)}}{2}\,2^{\frac {m+1}2}\ell.
\end{align*}
By the definition of the $g_i,\,s_i$, we have that $g_2(\aa)-g_1(\aa)+2^{k-1}(s_2(\aa)-s_1(\aa))=2^{k-2}$ in $\Z_{2^k}$.
Since
\begin{align*}
  \frac{(-1)^{\cc\cdot \iota^{-1}(g_1(\aa))+s_1(\aa)}-(-1)^{\cc\cdot \iota^{-1}(g_2(\aa))+s_2(\aa)}}{2}\in\{-1,0,1\},
\end{align*}
for the fixed $\aa\in\GF n$, the claim is proven.

Conversely, for a fixed $\aa\in\F_2^n$, assume that for all $\cc\in \F_2^{k-1}$, $f_\cc$ have their Walsh-Hadamard transforms of the form
  \[
  \Wa{f_\cc}(\aa)= \left((-1)^{\cc\cdot \iota^{-1}(g_1(\aa))+s_1(\aa)}-(-1)^{\cc\cdot \iota^{-1}(g_2(\aa))+s_2(\aa)}\right)\,2^{\frac {m-1}2}\ell,
   \]
   for some $g_j:\F_2^n\to \Z_{2^{k-1}}, s_j:\F_2^n\to\F_2$, $j=1,2$, with  $g_2(\aa)-g_1(\aa)+2^{k-1}(s_2(\aa)-s_1(\aa))=2^{k-2}$ in $\Z_{2^k}$, and $m,\ell$ odd integers. 

Observe that
\begin{align*}
 & \sum_{\cc\in\F_ 2^{k-1}}
   (-1)^{\cc\cdot (\dd\oplus \iota^{-1}(g_1(\aa))+s_1(\aa)}
-\sum_{\cc\in\F_ 2^{k-1}}  (-1)^{\cc\cdot (\dd\oplus \iota^{-1}(g_2(\aa))+s_2(\aa)}\\
  &= 
    \begin{cases}
    0 & \mbox{if $\iota(\dd)\not\in\{g_1(\aa),g_2(\aa)\}$}\\
    2^{k-1}(-1)^{s_1(\aa)} & \mbox{if $\iota(\dd)=g_1(\aa)\neq g_2(\aa)$}\\
    -2^{k-1}(-1)^{s_2(\aa)} & \mbox{if $\iota(\dd)=g_2(\aa)\neq g_1(\aa)$}\\
     2^{k-1}\left( (-1)^{s_1(\aa)}  -  (-1)^{s_2(\aa)}\right) & \mbox{if $\iota(\dd)=g_1(\aa)=g_2(\aa)$}.
    \end{cases} 
\end{align*}
  
    Further, using this identity and Lemma~\ref{lem:H-W}, we get
    \allowdisplaybreaks
\begin{align*}
 &\cH_{f} (\aa)
  = \frac{1}{2^{k-1}} \sum_{(\cc,\dd)\in\F_ {2}^{k-1}\times\F_{2}^{k-1}} (-1)^{\cc\cdot \dd}\zetak^{\iota(\dd)}
    \Wa{f_\cc}(\aa)\\
  &= 2^{\frac {m+1}{2}-k}\ell \sum_{\dd\in\F_2^{k-1}} \zetak^{\iota(\dd)} \left( \sum_{\cc\in\F_ 2^{k-1}}
   (-1)^{\cc\cdot (\dd\oplus \iota^{-1}(g_1(\aa))+s_1(\aa)}\right.\\
  &\qquad\qquad \qquad\qquad\qquad\qquad\qquad\left. -\sum_{\cc\in\F_ 2^{k-1}}  (-1)^{\cc\cdot (\dd\oplus \iota^{-1}(g_2(\aa))+s_2(\aa)}\right)\\
  &=  2^{\frac {m-1}{2}}\ell  \left( (-1)^{s_1(\aa)} \zetak^{g_1(\aa)}  -  (-1)^{s_2(\aa)}\zetak^{g_2(\aa)} \right) \\
  &=  2^{\frac {m-1}{2}}\ell\,    (-1)^{s_1(\aa)} \zetak^{g_1(\aa)} \left(1 -   \zetak^{g_2(\aa)-g_1(\aa)+2^{k-1}(s_2(\aa)-s_1(\aa))} \right)\\
  &= 2^{\frac {m-1}{2}} \ell\,   (-1)^{s_1(\aa)} \zetak^{g_1(\aa)} (1-\zetak^{2^{k-2}}) =2^{\frac{m}2}  \ell\,(-1)^{s_1(\aa)} \zetak^{g_1(\aa)}\bar\zeta_8,
\end{align*}
and so, $ |\cH_{f}(\aa)|=2^{\frac{m}2}\ell$. The claim is shown.

$(iv)$ First, let $\aa\in\supp(\cH_f)$, with $|\cH_f(\aa)|=2^{\frac{m}{2}}\ell$, $m,\ell\in\Z^{\geq 1}$ odd, $k=2$ (observe now that  $\zetak=i$), and assume that $\ell$ is the largest component of a Pythagorean triple, $\ell_1^2+\ell_2^2=\ell^2$.  By Theorem~\ref{reg-plat}, then $\cH_f(\aa)=2^{\frac{m-1}{2}}\left(\ell_1\epsilon_1+\ell_2\epsilon_2\pm  (\ell_1\epsilon_2-\ell_2\epsilon_1) \, i \right)$.

Using Lemma~\ref{lem:H-W}, we write
\begin{align*} 
  \cH_{f} (\aa)
  &= \frac{1}{2^{k-1}} \sum_{(\cc,\dd)\in\F_ {2}^{k-1}\times\F_{2}^{k-1}} (-1)^{\cc\cdot \dd}\zetak^{\iota(\dd)}
    \Wa{f_\cc}(\aa)\\
  &= \sum_{\dd\in\F_2^{k-1}} \left( \frac{1}{2^{k-1}} \sum_{\cc\in\F_ 2^{k-1}} (-1)^{\cc\cdot \dd}
    \Wa{f_\cc}(\aa) \right)\zetak^{\iota(\dd)}\\
&=\sum_{d\in\F_2} \left( \frac{1}{2} \sum_{c\in\F_ 2} (-1)^{cd}
    \Wa{f_c}(\aa) \right)i^{d}.
\end{align*}
Together with the previous identity, this renders 
$$\frac{1}{2} \sum_{c\in\F_ 2} (-1)^{cd}\Wa{f_c}(\aa)=2^{\frac{m-1}{2}}(\ell_1\epsilon_1+\ell_2\epsilon_2) \mbox{ for } d=0,$$
$$\frac{1}{2} \sum_{c\in\F_ 2} (-1)^{cd}\Wa{f_c}(\aa)=\pm2^{\frac{m-1}{2}}(\ell_1\epsilon_2-\ell_2\epsilon_1) \mbox{ for } d=1.$$

Thus,
\begin{align*}
  \Wa{f_c}(\aa)
  &= \frac{1}{2}  \sum_{(u,v)\in\F_ 2\times\F_ 2} (-1)^{(u + c) v}
    \Wa{f_u}(\aa)\\
  &=\sum_{v\in\F_ 2}(-1)^{cv}\frac{1}{2}
    \sum_{u\in\F_ 2} (-1)^{uv}
    \Wa {f_u}(\aa)\\
  &= 2^{\frac{m-1}{2}}\left(\ell_1\epsilon_1+\ell_2\epsilon_2\pm(-1)^c (\ell_1\epsilon_2-\ell_2\epsilon_1) \right).
\end{align*}
 Conversely, let $\Wa{f_c}(\aa)= 2^{\frac{m-1}{2}}\left(\ell_1\epsilon_1+\ell_2\epsilon_2\pm(-1)^c (\ell_1\epsilon_2-\ell_2\epsilon_1) \right)$. Then, using Lemma~\ref{lem:H-W}, we write
\begin{align*} 
  \cH_{f} (\aa)
&=\sum_{d\in\F_2} \left( \frac{1}{2} \sum_{c\in\F_ 2} (-1)^{cd}
    \Wa{f_c}(\aa) \right)i^{d}\\
&=\frac{1}{2}2^{\frac{m-1}{2}}\sum_{d\in\F_2} \left( (-1)^0(\ell_1\epsilon_1+\ell_2\epsilon_2\pm(\ell_1\epsilon_2-\ell_2\epsilon_1))\right.\\
&\qquad\qquad\qquad \left.+(-1)^d(\ell_1\epsilon_1+\ell_2\epsilon_2\mp(\ell_1\epsilon_2-\ell_2\epsilon_1))\right)i^d\\
&=2^{\frac{m-1}{2}}(\ell_1\epsilon_1+\ell_2\epsilon_2)\pm(\ell_1\epsilon_2-\ell_2\epsilon_1)i.
\end{align*}

Finally, let $\aa\in\supp(\cH_f)$, with $|\cH_f(\aa)|=2^{\frac{m}{2}}\ell$, $m,\ell\in\Z^{\geq 1}$ odd, $k=2$ and $\ell$ is not the largest component of a Pythagorean triple.  As before, from Theorem~\ref{reg-plat}, then $\cH_f(\aa)=2^{\frac{m-1}{2}}\ell\,(\pm 1\pm i)$.

Using Lemma~\ref{lem:H-W}, we write
\begin{align*} 
  \cH_{f} (\aa)
  &= \frac{1}{2^{k-1}} \sum_{(\cc,\dd)\in\F_ {2}^{k-1}\times\F_{2}^{k-1}} (-1)^{\cc\cdot \dd}\zetak^{\iota(\dd)}
    \Wa{f_\cc}(\aa)\\
  &= \sum_{\dd\in\F_2^{k-1}} \left( \frac{1}{2^{k-1}} \sum_{\cc\in\F_ 2^{k-1}} (-1)^{\cc\cdot \dd}
    \Wa{f_\cc}(\aa) \right)\zetak^{\iota(\dd)}\\
&=\sum_{d\in\F_2} \left( \frac{1}{2} \sum_{c\in\F_ 2} (-1)^{cd}
    \Wa{f_c}(\aa) \right)i^{d}.
\end{align*}
Together with the previous value of $\cH_f(\aa)$, this renders 
$$\frac{1}{2} \sum_{c\in\F_ 2} (-1)^{cd}\Wa{f_c}(\aa)=\pm2^{\frac{m-1}{2}}\ell\ \mbox{ for } d=0,1.$$
Thus,
\begin{align*}
  \Wa{f_c}(\aa)
  &= \frac{1}{2}  \sum_{(u,v)\in\F_ 2\times\F_ 2} (-1)^{(u + c) v}
    \Wa{f_u}(\aa)
  =\sum_{v\in\F_ 2}(-1)^{cv}\frac{1}{2}
    \sum_{u\in\F_ 2} (-1)^{uv}
    \Wa {f_u}(\aa)\\
  &= 2^{\frac{m-1}{2}}\ell\,(\pm1\pm(-1)^c).
\end{align*}
 Conversely, let $\Wa{f_c}(\aa)= 2^{\frac{m-1}{2}}\ell\,(\pm1\pm(-1)^c)$. Then, using Lemma~\ref{lem:H-W}, we write
\begin{align*} 
  \cH_{f} (\aa)
&=\sum_{d\in\F_2} \left( \frac{1}{2} \sum_{c\in\F_ 2} (-1)^{cd}
    \Wa{f_c}(\aa) \right)i^{d}\\
&=\frac{1}{2}2^{\frac{m-1}{2}}\ell\sum_{d\in\F_2} \left( (-1)^0(\pm1\pm1)+(-1)^d(\pm1\mp1)\right)i^d\\
&=2^{\frac{m-1}{2}}\ell\,(\pm1\pm i),
\end{align*}
proving the last claim.
\end{proof}

The following corollary is then immediate.
\begin{corollary}
 Let $f:\F_2^n\rightarrow\Z_{2^k}$,  $k\geq 1$, be a  function given as
$f(\xx) = a_0(\xx)+2a_1(\xx)+\cdots+ 2^{k-1}a_{k-1}(\xx)$.  Let $s\geq 0$ be an integer.
Then $f$ is $s$-plateaued  if and only if for each $\cc\in\F_2^{k-1}$, the Boolean function $f_\cc$
defined as in Theorem~\textup{\ref{maj-thm}} 
is an $s$-plateaued (if $n+s$ is even), respectively, an $(s+1)$-plateaued function (if $n+s$ is odd) with   the extra conditions on the Walsh-Hadamard coefficients, as in Theorem~\textup{\ref{maj-thm}}.
\end{corollary}
We will derive other characterizations of plateaued functions in the last section.

\section{Some constructions of generalized landscape functions}
  
 First, we start with some examples of five valued spectra   (certainly, landscape) functions, and later on, we shall give constructions of arbitrary length generalized landscape functions.  In~\cite[Theorem 19]{MS02}, a class of functions with five valued spectra was constructed. These  are  $n$-variable, $m$-resilient, degree $(n - m - 1)$ functions with nonlinearity 
$ nl(f)=2^{n-1} - 2^{(n+m-2)/2}$, if $n - m + 1$ is odd,  and
$nl(f)=2^{n-1} - 2^{(n+m-1)/2}$, if $n - m + 1$ is even, with 
five valued Walsh spectrum (under $n-m\geq 5$), namely, $\{\pm 2^{(n+m)/2},2^{(n+m)/2}-2^{m+2},-2^{m+2},0\}$, if $n-m+1$ is odd, respectively, $\{\pm 2^{(n+m+1)/2},2^{(n+m+1)/2}-2^{m+2},-2^{m+2},0\}$, if $n-m+1$ is even (see~\cite{CH1,CS17} for these definitions). To generate landscape  functions that have five valued spectra, we take $m:=n-5$, and so, $n-m+1=6$, $n+m+1=2n-4$, and the spectrum will be $\{ \pm 2^{n-2}, \pm 2^{n-3},0\}$; also, $n-m=6$, so $n-m+1=7$, and the spectrum will be $\{ \pm 2^{n-3}, \pm 2^{n-4},0\}$. 

We do not need it here, but using Catalan's Conjecture (now known as Mih\u ailescu's Theorem~\cite{Mi04}, which states that the only nontrivial (that is, $a,b>1,x,y>0$) diophantine solution to $x^a-y^b=1$ is $x=3,a=2,y=2,b=3$),
we can infer that we can only get these examples of landscape functions from the  specific  construction of~\cite[Theorem 19]{MS02}. 
 
It is not very difficult to show that landscape functions exist for every dimension, as our next proposition shows, which adapts some classical inductive plateaued construction (see, for instance, \cite{KSW85,MMS17,gbps,KUS1,smgs} for the construction of generalized Boolean bent functions), as well as the paper~\cite{Si14}, which contains some constructions of semibent and even more general plateaued in the spirit of Maiorana-McFarland construction of bent functions.
\begin{proposition}
Let $a\in\F_2$, $q=2^k$, $k\geq 1$, $f$ be a generalized Boolean function in $\cGB_n^q$ and $g:\F_2^{n+1}\to \Z_q$ be defined by $g(\xx,y)=f(\xx)+2^{k-1}ay$.  If $f$ is $s$-plateaued, then $g$ is $(s+1)$-plateaued (in particular, if $f$ is generalized bent, then $g$ is $1$-plateaued). 
In general, if $f$ is a landscape function of length $t$ and levels  $\{(m_1,\ell_1),(m_2,\ell_2),\ldots\}$, then $g$ is a landscape function of length $t$ and levels  $\{(m_1+1,\ell_1),(m_2+1,\ell_2),\ldots\}$.
\end{proposition}
\begin{proof}
Let $f$ be a landscape function of  levels  $\{(m_1,\ell_1),(m_2,\ell_2),\ldots, (m_t,\ell_t)\}$.
We compute the Walsh-Hadamard transform of $g$ at $(\uu,\vv)\in\F_2^n\times \F_2$, and get
\begin{align*}
\cH_g(\uu,\vv)&=\sum_{(\xx,y)\in\F_2^n\times \F_2} \zeta^{f(\xx)+2^{k-1}ay}(-1)^{\uu\cdot \xx+vy}\\
&=\sum_{\xx\in\F_2^n,y=0} \zeta^{f(\xx)}(-1)^{\uu\cdot \xx}+\sum_{\xx\in\F_2^n,y=1} \zeta^{f(\xx)+2^{k-1}a}(-1)^{\uu\cdot \xx+v}\\\
&=\cH_f(\uu)+(-1)^{a+v}  \cH_f(\uu)=\left(1+(-1)^{a+v}  \right)\cH_f(\uu).
\end{align*}
Thus, $|\cH_g(\uu,v)|\in\{0,2|\cH_f(\uu)|\}$,
from which we can infer all of our claims.
\end{proof}

 Carlet~\cite{Car03} introduced a secondary construction (often called the ``indirect sum''),
  as follows. Let $n =r+s$, where $r$ and $s$ are positive integers, and $f_1, f_2\in\cB_r$, $g_1, g_2\in\cB_s$. Define $h$ as   the concatenation of the four functions $f_1$, $\bar f_1$,
$f_2$, $\bar f_2$, in an order controlled by $g_1(\yy)$ and $g_2(\yy)$,
 \begin{equation}
 \label{eq-indirect}
 h(\xx,\yy)=f_1(\xx)\+g_1(\yy)\+(f_1(\xx)\+f_2(\xx))(g_1(\yy)\+g_2(\yy)).
 \end{equation}
 It is known that in the Boolean case, if $r,s$ are even and $f_1,f_2$ are semibent and $g_1,g_2$ are bent, then $h$ is semibent. In fact, a more general result is true as we shall show next.
 The following lemma is known and easy to show.
 \begin{lemma}
 \label{lem-zs}
 For $s\in\BBF_2$ and $z\in\mathbb{C}$, we have
\begin{equation*}
\label{eq-zs}
z^s=\frac{1+(-1)^s}{2}+\frac{1-(-1)^s}{2}z.
\end{equation*}
 \end{lemma}
 \begin{thm}
 \label{constr-ind}
 Let $q=2^k$, $k\geq 1$, and $h:\F_2^r\times\F_2^s\to \Z_q$ be given by $h(\xx,\yy)=f_1(\xx)+2^{k-1}g_1(\yy)+(f_2(\xx)-f_1(\xx))(g_1(\yy)+g_2(\yy))$ (all operations are in $\Z_q$), where $f_1, f_2\in\cGB_r^q$, $g_1, g_2\in\cB_s$, $q=2^k$, with $g_1,g_2$ bent (thus, $s$ is even). The following hold:
 \begin{enumerate}
 \item[$(i)$] If $f_1,f_2$ are $t$-plateaued, then $h$ is $t$-plateaued (hence,  of length $2$). 
 \item[$(ii)$] If  $t_1\neq t_2$ and $g_1\neq g_2, g_1\neq \overline g_2$, then $h$ is a landscape function of length $3$, namely, the moduli of its spectra are $\{0,2^{\frac{n+t_1}{2}}, 2^{\frac{n+t_2}{2}}\}$.  In particular, if $q=2$, then $h$ has five valued spectra.
 \item[$(iii)$]  If $f_1,f_2$ are landscape functions such that
    \begin{align*}
  \{|\cH_{f_1}(\uu)|\}_{\uu\in\supp(\cH_{f_1})}&=\{2^{\frac{p_1}{2}}\ell_{11}^2,\ldots,2^{\frac{p_t}{2}}\ell_{1t}^2\},\\
  \{|\cH_{f_2}(\uu)|\}_{\uu\in\supp(\cH_{f_2})}&=\{2^{\frac{q_1}{2}}\ell_{21}^2,\ldots,2^{\frac{q_f}{2}}\ell_{2f}^2\},
  \end{align*}
  and $g_1\neq g_2, g_1\neq \overline g_2$,
  then  $h$ is a landscape function such that
  \[
    \{|\cH_h(\uu,\vv)|\}_{\uu,\vv}=\{0,2^{\frac{s+p_1}{2}}\ell_{11}^2,\ldots,2^{\frac{s+p_t}{2}}\ell_{1t}^2,2^{\frac{s+q_1}{2}}\ell_{21}^2,\ldots,2^{\frac{s+q_f}{2}}\ell_{2f}^2\}.
  \]
 \end{enumerate}
 \end{thm}
 \begin{proof}
 Using Lemma~\ref{lem-zs} and the indicators of the sets $\{\yy\,:\, g_1(\yy)+g_2(\yy)=j\}$, $j=0,1$, being $\displaystyle \frac{1+(-1)^j(-1)^{g_1(\yy)+g_2(\yy)}}{2} $, we compute the Walsh--Hadamard transform of $h$ and obtain 
 \allowdisplaybreaks
 \begin{align*}
 2\cH_h(\uu,\vv)
 &=2\sum_{\xx\in\F_2^r,\yy\in\F_2^s} \zeta^{f_1(\xx)+2^{k-1}g_1(\yy)+(f_2(\xx)-f_1(\xx))(g_1(\yy)+g_2(\yy))}(-1)^{\uu\cdot \xx+\vv\cdot \yy}\\
 &= 2\sum_{\xx\in\F_2^r}  \sum_{\substack{\yy\in\F_2^s\\g_1(\yy)+g_2(\yy)=0}} \zeta^{ f_1(\xx)+2^{k-1}g_1(\yy)} (-1)^{\uu\cdot \xx+\vv\cdot \yy}\\
 &\qquad + 2\sum_{\xx\in\F_2^r}  \sum_{\substack{\yy\in\F_2^s\\g_1(\yy)+g_2(\yy)=1}}  \zeta^{ f_2(\xx)+2^{k-1}g_1(\yy)} (-1)^{\uu\cdot \xx+\vv\cdot \yy}\\
  &= 2\sum_{\xx\in\F_2^r}  \zeta^{ f_1(\xx)} (-1)^{\uu\cdot \xx} \sum_{\yy\in\F_2^s} (-1)^{g_1(\yy)} \frac{1+(-1)^{g_1(\yy)+g_2(\yy)}}{2} (-1)^{\vv\cdot \yy}\\
 &\qquad + 2\sum_{\xx\in\F_2^r}  \zeta^{ f_2(\xx)} (-1)^{\uu\cdot \xx} \sum_{\yy\in\F_2^s} (-1)^{g_1(\yy)} \frac{1-(-1)^{g_1(\yy)+g_2(\yy)}}{2} (-1)^{\vv\cdot \yy}\\
&= \cH_{f_1}(\uu)   \sum_{\yy\in\F_2^s}\left( (-1)^{g_1(\yy)}+(-1)^{g_2(\yy)} \right) (-1)^{\vv\cdot \yy}\\
&\qquad\qquad\qquad +\cH_{f_2}(\uu)   \sum_{\yy\in\F_2^s}\left( 
(-1)^{g_1(\yy)}-(-1)^{g_2(\yy)} \right) (-1)^{\vv\cdot \yy}\\
&=\cH_{f_1}(\uu) \left(\cW_{g_1}(\vv)+\cW_{g_2}(\vv) \right)
+\cH_{f_2}(\uu) \left(\cW_{g_1}(\vv)-\cW_{g_2}(\vv) \right).
  \end{align*}
  If $f_1,f_2$ are, respectively, $t_1$, $t_2$-plateaued, and $g_1,g_2$ are bent, and observing that 
  \[
  \left(\cW_{g_1}(\vv)+\cW_{g_2}(\vv) \right)\left(\cW_{g_1}(\vv)-\cW_{g_2}(\vv) \right)=\cW_{g_1}^2(\vv)-\cW_{g_2}^2(\vv) =0,
  \]
  then either $\left(\cW_{g_1}(\vv)+\cW_{g_2}(\vv) \right)=\pm 2^{\frac{s}{2}+1}$ and $\left(\cW_{g_1}(\vv)-\cW_{g_2}(\vv) \right)=0$, rendering $|\cH_h(\uu,\vv)|\in\left\{0,2^{-1} 2^{\frac{r+t_1}{2}}\cdot 2^{\frac{s}{2}+1}\right\}=\left\{0,2^{\frac{n+t_1}{2}}\right\}$, or $\left(\cW_{g_1}(\vv)+\cW_{g_2}(\vv) \right)=0$ and $\left(\cW_{g_1}(\vv)-\cW_{g_2}(\vv) \right)=\pm 2^{\frac{s}{2}+1}$, rendering $|\cH_h(\uu,\vv)|\in \left\{0,2^{-1} 2^{\frac{r+t_2}{2}}\cdot 2^{\frac{s}{2}+1}\right\}= \left\{0,2^{\frac{n+t_2}{2}}\right\}$. If $t_1\neq t_2$, we have the cases (both ocurring, since $g_1\neq g_2,\overline g_2$), $|\cH_h(\uu,\vv)|\in \left\{0,2^{\frac{n+t_1}{2}}\right\}$ and  $|\cH_h(\uu,\vv)|\in \left\{0,2^{\frac{n+t_2}{2}}\right\}$.
  Thus, $|\cH_h(\uu,\vv)|\in\left\{0,2^{\frac{n+t_1}{2}},2^{\frac{n+t_2}{2}}\right\}$. Certainly, if $t_1= t_2=t$, then $|\cH_h(\uu,\vv)|\in\left\{0,2^{\frac{n+t}{2}}\right\}$, so $h$ is $t$-plateaued. Claim $(i)$ and $(ii)$ are shown.
  
  Next, assume that $f_1,f_2$ are landscape functions, and pick $\uu\in \supp(\cH_{f_1})\cap  \supp(\cH_{f_2})$, and so,  $|\cH_{f_1}(\uu)|=2^{\frac{p}{2}}\ell_1$, $|\cH_{f_2}(\uu)|=2^{\frac{q}{2}}\ell_2$, for some $p,q\in\Z$ and odd $\ell_1,\ell_2$. The argument above shows that either $|\cH_h(\uu,\vv)|\in \left\{0,2^{\frac{s+p}{2}}\ell_1\right\}$, or  $|\cH_h(\uu,\vv)|\in \left\{0,2^{\frac{s+q}{2}}\ell_2\right\}$ (both occurring since $g_1\neq g_2, g_1\neq \overline g_2$). If $\uu\in \supp(\cH_{f_1})\, \cap \, \overline{\supp(\cH_{f_2})}$ (respectively, $\uu\in \overline{\supp(\cH_{f_1})}\,\cap\,  \supp(\cH_{f_2})$), then $|\cH_h(\uu,\vv)|\in \left\{0,2^{\frac{s+p}{2}}\ell_1\right\}$ (respectively, $|\cH_h(\uu,\vv)|\in \left\{0,2^{\frac{s+q}{2}}\ell_2\right\}$). If $\uu\in \overline{\supp(\cH_{f_1})}\,\cap  \,\overline{\supp(\cH_{f_2})}$, then $|\cH_h(\uu,\vv)|=0$. Therefore, if
    \begin{align*}
  \{|\cH_{f_1}(\uu)|\}_{\uu\in\supp(\cH_{f_1})}&=\{2^{\frac{p_1}{2}}\ell_{11},\ldots,2^{\frac{p_t}{2}}\ell_{1t}\},\\
  \{|\cH_{f_2}(\uu)|\}_{\uu\in\supp(\cH_{f_2})}&=\{2^{\frac{q_1}{2}}\ell_{21},\ldots,2^{\frac{q_f}{2}}\ell_{2f}\},
  \end{align*}
  then 
  $\displaystyle
  \{|\cH_h(\uu,\vv)|\}_{\uu,\vv}=\{0,2^{\frac{s+p_1}{2}}\ell_{11},\ldots,2^{\frac{s+p_t}{2}}\ell_{1t},2^{\frac{s+q_1}{2}}\ell_{21},\ldots,2^{\frac{s+q_f}{2}}\ell_{2f}\}.
  $
\end{proof}

\section{Characterizing plateaued functions in terms of  second derivatives and fourth moments}

 \begin{thm}
 Let $f:\F_2^n\to\Z_{2^k}$, $k\geq 2$, $s$ be an  integer with $0\leq s\leq n$, and $\zeta:=\zeta_{2^k}=e^{\frac{2\pi\, i}{2^k}}$ be the primitive root of $1$. Then $f$ is $s$-plateaued if and only if for all $\xx\in\F_2^n$,
 \[
 \sum_{\aa,\bb\in\F_2^n} \zeta^{D_\bb D_\aa f(\xx)}=2^{n+s}.
 \]
 Furthermore,   $f$ is $s$-plateaued if and only if  
\[
\sum_{\dd   \in \BBF_2^n}  |\cH_f(\dd)|^4=2^{3n+s}.
\]
 \end{thm}
 \begin{proof}
Let $\xx\in\F_2^n$ be fixed.  
First observe that  
 \begin{align*}
 & \sum_{\aa,\bb\in\F_2^n} \zeta^{D_\bb D_\aa f(\xx)} =2^{n+s} \text{ is equivalent to}\\
&F_1(\xx):= \sum_{\aa,\bb\in\F_2^n} \zeta^{f(\xx\+\aa\+\bb)-f(\xx\+\bb)-f(\xx\+\aa)} =2^{n+s}\zeta^{-f(\xx)}=:F_2(\xx),
   \end{align*}
   which is further equivalent to their Fourier transforms being equal at all $\uu\in\F_2^n$, that is,
       \begin{align}
       \label{eq:id}
 \sum_{\xx,\aa,\bb\in\F_2^n} \zeta^{f(\xx\+\aa\+\bb)-f(\xx\+\bb)-f(\xx\+\aa)} (-1)^{\uu\cdot \xx}=2^{n+s}\sum_{\xx\in\F_2^n}\zeta^{-f(\xx)} (-1)^{\uu\cdot \xx}.
   \end{align}
   We compute the two expressions in~\eqref{eq:id}, separately. We start with the left hand side of~\eqref{eq:id}, and setting $\aa_1:=\xx\+\aa,\bb_1:=\xx\+\bb$, we obtain
   \allowdisplaybreaks
 \begin{align*}
 &\sum_{\xx,\aa,\bb\in\F_2^n} \zeta^{f(\xx\+\aa\+\bb)-f(\xx\+b)-f(\xx\+\aa)} (-1)^{u\cdot \xx}\\
 =&  \sum_{\xx,\aa_1,\bb_1\in\F_2^n} \zeta^{f(\xx\+\aa_1\+\bb_1)-f(\bb_1)-f(\aa_1)} (-1)^{\uu\cdot \xx}\\
 =&   \sum_{\bb_1\in\F_2^n} \zeta^{ -f(\bb_1) } (-1)^{\uu\cdot \bb_1}\
   \sum_{\aa_1\in\F_2^n} \zeta^{ -f(\aa_1)} (-1)^{\uu\cdot \aa_1}\\
   &\qquad\qquad\qquad\cdot \sum_{\xx\in\F_2^n} \zeta^{f(\xx\+\aa_1\+\bb_1) } (-1)^{\uu\cdot (\xx\+\bb_1\+\aa_1)}\\
  =&\overline{\cH_f(\uu)}\ \overline{\cH_f(\uu)}\,\cH_f(\uu) =|\cH_f(\uu)|^2\, \overline{\cH_f(\uu)}.
    \end{align*}
   The right hand side of~\eqref{eq:id} can be written as
  \begin{align*}  
   2^{n+s}\sum_{\xx\in\F_2^n}\zeta^{-f(\xx)} (-1)^{\uu\cdot \xx}
   =&2^{n+s}\overline{\cH_f(\uu)},
     \end{align*} 
     therefore \eqref{eq:id} is equivalent to
  $
|\cH_f(\uu)|^2\, \overline{\cH_f(\uu)}  =2^{n+s}\,\overline{\cH_f(\uu)},
 $
     that is, $|\cH_f(\uu)|\in\{0,2^{(n+s)/2}\}$. Our first claim is shown.
     
     Next,  using \cite[Theorem 1]{smgs} ({\em caution}: our generalized Walsh-Hadamard transform is not normalized), we see that
     \allowdisplaybreaks
  \begin{align*}  
  \sum_{\aa,\bb\in\F_2^n} \zeta^{D_bD_a f(\xx)}
  &=  \sum_{\aa,\bb\in\F_2^n} \zeta^{f(x\+\aa\+\bb)-f(x\+\bb)-f(\xx\+\aa)+f(\xx)}\\
  &= \sum_{\aa\in\F_2^n}  \zeta^{f(\xx) -f(\xx\+\aa)}\ \sum_{\bb\in\F_2^n}  \zeta^{f(\xx\+\aa\+\bb)-f(\xx\+\bb)}\\
  & {\substack{\cc:=\xx\+\bb\\\downarrow}\atop {=}}\sum_{\aa\in\F_2^n}  \zeta^{f(\xx) -f(\xx\+\aa)}\ \sum_{\cc\in\F_2^n}  \zeta^{f(\cc\+\aa)-f(\cc)}\\
 & = \sum_{\aa\in\F_2^n}  \zeta^{f(\xx) -f(\xx\+\aa)}  \cC_f(\aa),\text{ since $ \cC_f$ is always real}\\
 &=2^{-n}\sum_{\aa\in\F_2^n}  \zeta^{f(\xx) -f(\xx\+\aa)}   \sum_{\dd   \in \BBF_2^n}|\cH_f(\dd)|^2 (-1)^{\aa  \cdot \dd }\\
  & = 2^{-n}\sum_{\dd   \in \BBF_2^n}  |\cH_f(\dd)|^2 \sum_{\aa\in\F_2^n}  \zeta^{f(\xx) -f(\xx\+\aa)}  (-1)^{\aa  \cdot \dd }\\
   &=2^{-n}  \sum_{\dd   \in \BBF_2^n}  |\cH_f(\dd)|^2 \zeta^{f(\xx)}  (-1)^{\xx\cdot \dd} \sum_{\cc\in\F_2^n}  \zeta^{-f(\cc)}  (-1)^{\cc \cdot \dd }\\
   &= 2^{-n}\sum_{\dd   \in \BBF_2^n}  |\cH_f(\dd)|^2 \zeta^{f(\xx)} (-1)^{\xx\cdot \dd} \overline{\cH_f(\dd)}.
 \end{align*}
 Thus,
 \begin{align*}
  2^{2n+s}&=\sum_{\xx,\aa,\bb\in\F_2^n} \zeta^{D_\bb D_\aa f(\xx)}
  =2^{-n} \sum_{\xx,\dd   \in \BBF_2^n}  |\cH_f(\dd)|^2  \overline{\cH_f(\dd)} \zeta^{f(\xx)} (-1)^{\xx\cdot \dd}\\
  &=2^{-n} \sum_{\dd   \in \BBF_2^n}  |\cH_f(\dd)|^2  \overline{\cH_f(\dd)} \sum_{\xx   \in \BBF_2^n}  \zeta^{f(\xx)} (-1)^{\xx\cdot \dd}
  =2^{-n} \sum_{\dd   \in \BBF_2^n}  |\cH_f(\dd)|^4,
 \end{align*}
 and the second claim is shown.
\end{proof}
Our next corollary (see~\cite{Car03} for the classical counterpart) is immediate since a generalized bent function corresponds to a $0$-plateaued function.
\begin{corollary}  
A function $f\in\cGB_n^q$, $q=2^k$, is generalized bent if and only if $  \sum_{\aa,\bb\in\F_2^n} \zeta^{D_\bb D_\aa f(\xx)} =2^{n}$ if and only if  $\sum_{\dd   \in \BBF_2^n}  |\cH_f(\dd)|^4=2^{3n}$.
\end{corollary}

\noindent{\bf Acknowledgement.} The paper was started while the second author visited Western Norway University of Applied Sciences. This author thanks the institution for the excellent working conditions while the paper was being written.

\end{document}